\documentclass[12pt]{article}
\usepackage[left=2.5cm, right=2.5cm, top=2.5cm, bottom=2.5cm]{geometry}
\usepackage{amsmath,amsthm,amssymb,natbib,xcolor,enumitem}
\usepackage{enumitem}
\usepackage[hyphens]{url}
\usepackage{hyperref}
\hypersetup{
    colorlinks,
    linkcolor={red!50!black},
    citecolor={black},
    urlcolor={black}
}
\setlist{nolistsep}

\newtheorem{thm}{Theorem}
\newtheorem{prop}[thm]{Proposition} 
\newtheorem{lemma}[thm]{Lemma}

\theoremstyle{definition}

\theoremstyle{remark}

\newtheorem{example}{Example}

\newcommand*{\set}[1]{\left\{#1\right\}}
\newcommand*{\act}{a} 
\newcommand*{\lact}{L} 
\newcommand*{\ract}{R} 
\newcommand*{\lstate}{L} 
\newcommand*{\rstate}{R} 
\newcommand*{\lsign}{\ell} 
\newcommand*{\rsign}{r} 
\newcommand*{\Lsign}{L} 
\newcommand*{\Rsign}{R}
\newcommand*{\losign}{\lambda} 
\newcommand*{\rosign}{\rho} 
\newcommand*{\cost}{k}

\title{Herding driven by the desire to differ}
\author{Sander Heinsalu\thanks{Research School of Economics, Australian National University, 
25a Kingsley St, Acton ACT 2601, Australia.
Email: sander.heinsalu@anu.edu.au, 
website: \url{https://sanderheinsalu.com/}
}}
\date{\today}

\begin{document}
\maketitle

\begin{abstract}
Observational learning often involves congestion: an agent gets lower payoff from an action when more predecessors have taken that action. This preference to act differently from previous agents may paradoxically increase all but one agent's probability of matching the actions of the predecessors. 
The reason is that when previous agents conform to their predecessors despite the preference to differ, their actions become more informative. 
The desire to match predecessors' actions may reduce herding by a similar reasoning. 

	
\end{abstract}

This paper studies rational agents' learning from the choices of others when the information of others is not directly available. Payoffs are interdependent due to congestion costs: if more preceding agents choose an action, then an agent's payoff from taking that action falls. 
Congestion arises in many situations, such as when individuals choose a supermarket lane or other queue, a route to drive or a service provider to use. For firms, choosing a market that others have entered is less profitable, other things equal. 

The model follows the seminal papers of \cite{banerjee1992} and \cite{bikhchandani+1992} on herding and information cascades. Agents choose in sequence between two actions, after observing the previous agents' actions and a private signal. All agents prefer their action to match a binary state, which is symmetrically unknown. The payoff of an agent also increases when the action differs from those of the preceding agents. 
Such social preferences are also assumed in \cite{gaigl2009} and \cite{eyster+2014}. 

In equilibrium, the preference for an action different from that of previous agents may paradoxically increase all but one agent's probability of matching the actions of the predecessors, compared to the case when payoffs do not depend on others' actions. 
When previous agents choose the same action as their predecessors, congestion costs increase the informativeness of this action. The reason is that a stronger signal is required to induce an action when the preceding agents have chosen it. A more informative action in turn motivates imitation, even when congestion moderately increases the cost of imitating. 

Similarly, the desire to conform to previous agents' actions may reduce herding. If past agents made the same choice as their predecessors, then these actions are less informative under a preference to match previous moves. The decreased informativeness of preceding actions allows an agent's private signal to outweigh the combined effect of previous moves and the desire to conform. 

In contrast to the current work, both \cite{gaigl2009} and \cite{eyster+2014} show that congestion costs reduce herding and, if not too large, improve learning. Large enough congestion costs cause agents to alternate their actions (anti-herd), which decreases learning. \cite{gaigl2009} and \cite{eyster+2014} focus on asymptotic learning, but the present paper considers the probability of each agent matching the actions of his predecessors, as well as the correct action. 
As in the previous literature, when the desire to differ is small enough, all agents take the same action from some finite time onward. In that case, learning is bounded, i.e.\ there is positive probability of the wrong action as time goes to infinity. 

In \cite{callander+horner2009}, agents are differently informed and observe only the number of previous movers choosing an action, not who chose it. Following the minority is sometimes optimal. 

Other forms of social preference in herding have been studied. In \cite{ali+kartik2012}, agents prefer others to take the correct action. \cite{callander2007} assumes that agents want to match the eventual majority, thus payoff depends on future agents' choices, unlike in the current work. 

The next section sets up the model where agents desire to differ from previous movers. The results are collected in Section~\ref{sec:results} and discussed in Section~\ref{sec:discussion}. The appendix shows that a desire to conform may reduce herding.

\section{Model}
\label{sec:model}

Time is discrete, with periods and players indexed by $i\in\mathbb{N}$. In period $i$, player $i$ observes a private signal $s_{i}\in\set{\Lsign,\lsign,\rsign,\Rsign}$ and chooses a public action $\act_{i}\in\set{\lact,\ract}$. The public history of actions up to time $t$ is denoted $\act^{t}=(\act_{1},\ldots,\act_{t})$. 
Action $\act_{i}$ is called \emph{uninformative} after history $\act^{i-1}$ if $\act_{i}(\act^{i-1},s_{i})$ is constant in $s_{i}$, and \emph{informative} otherwise. An \emph{information cascade} is said to occur after history $\act^{t}$ if actions $a_{i}$, $i> t$ are uninformative after any continuation history $\act^{i}=(\act^{t},\act_{t+1},\ldots)$. 
\emph{Herding} after history $\act^{t}$ means that $\act_{t+1}=\act_{t}$ regardless of signals. Thus a herd is a special case of an information cascade. 

An unknown state $\theta\in\set{\lstate,\rstate}$ determines payoffs via 
\begin{align*} 
u_{i}(\act^{i},\theta) =\text{\textbf{1}}\set{\act_{i} =\theta} -\frac{\cost}{i-1}\sum_{j=1}^{i-1}\text{\textbf{1}}\set{\act_{j} =\act_{i}},
\end{align*} 
where \textbf{1}$S$ denotes the indicator function of set $S$ and $\cost\geq 0$ is the congestion cost. 
If $\cost=0$, then the environment is standard herding with independent preferences. If $\cost>0$, then each player prefers to take a different choice than the majority of the previous agents, other things equal. 

The prior probability of state $\rstate$ is $p_{0}\in[\frac{1}{2},1)$ w.l.o.g. 
Denote by $p_S\in(0,1)$ the unconditional probability of signal $s_{i}\in\set{\Lsign,\Rsign}$. Conditional on the state, the probabilities of the signals are $\Pr(\Lsign|\lstate) =\Pr(\Rsign|\rstate) =Q\in\left(\frac{p_S}{2},p_S\right)$ and $\Pr(\lsign|\lstate) =\Pr(\rsign|\rstate) =q\in\left(\frac{1-p_S}{2},\frac{Q(1-p_S)}{p_S}\right)$. 
Therefore $\Pr(\Lsign|\rstate) =\Pr(\Rsign|\lstate) =p_S-Q$ and $\Pr(\lsign|\rstate) =\Pr(\rsign|\lstate) =1-p_S-q$. 
Bayesian updating determines each player's posterior belief $p_i=\Pr(\rstate|\act^{i-1},s_{i})$ and log likelihood ratio $l_i:=\ln p_i-\ln(1-p_i)$. Using $l_i$ instead of $p_i$ simplifies the exposition and is mathematically equivalent. 
Signals $\Lsign,\lsign$ favour state $\lstate$, in the sense of increasing the posterior probability of $\lstate$. Similarly, $\Rsign,\rsign$ favour state $\rstate$. 
Calling signals $\lsign,\rsign$ \emph{weak} and $\Lsign,\Rsign$ \emph{strong} is justified by 
$\frac{q}{1-p_S}<\frac{Q}{p_S}$, which means that the posterior belief moves more in response to $\Lsign,\Rsign$ than to $\lsign,\rsign$. 
Assume $q>p_{0}(1-p_S)$, 
equivalently $l_q>l_0$, to ensure signals are informative enough for even a weak signal $s\in\set{\lsign,\rsign}$ to overturn the prior, i.e.\ player $1$ to believe after signal $\lsign$ that state $\lstate$ is more likely than $\rstate$. 

Denote the (public) log likelihood ratio of player $i>1$ before observing $s_{i}$ by $l_{i}(\act^{i-1})$.
Action $\tilde{a}_{i}\in\set{\lact,\ract}$ is called \emph{more informative} than $\hat{a}_{j}\in\set{\lact,\ract}$ if $|l_{i+1}((a^{i-1},\tilde{a}_{i}))-l_{i}(a^{i-1})| \geq |l_{j+1}((a^{j-1},\hat{a}_{j}))-l_{j}(a^{j-1})|$ for any $a^{i-1},a^{j-1}$, which means that $\tilde{a}_{i}$ moves the public log likelihood ratio $l_{i+1}$ more than $\hat{a}_{j}$ moves $l_{j+1}$. 

To derive player $i$'s private log likelihood ratios $l_i(\act^{i-1},s_i)$ after $a^{i-1},s_i$, define 
\begin{description}
\item $l_Q:=\ln Q-\ln (p_S-Q) >0$, 
\item $l_q:=\ln q-\ln (1-p_S-q) \in(0,l_Q)$, 
\item $l_{Qq}:= \ln(Q+q)-\ln(1-Q-q)\in(l_q,l_Q)$ and 
\item $l_{\neg Q}:= \ln(1-p_S+Q)-\ln(1-Q)\in(0,l_{Qq})$, 
\end{description}
where $l_Q,l_q$ are the log likelihood ratios of strong and weak signals respectively. The log likelihood ratio $l_{Qq}$ does not distinguish strong and weak signals, only whether the signal favours $\lstate$ or $\rstate$. If the strong signal in favour of one state is distinguishable from the other three, but the latter look identical to an agent, then upon not seeing the distinguishable strong signal, the agent uses $l_{\neg Q}$ to update. 
The (private) log likelihood ratios of $i$ upon observing $s_{i}$ are 
\begin{align*}
&l_{i}(\act^{i-1},\Lsign) =l_{i}(\act^{i-1})-l_Q, 
\qquad l_{i}(\act^{i-1},\lsign) =l_{i}(\act^{i-1})-l_q, 
\\&\notag l_{i}(\act^{i-1},\Rsign) =l_{i}(\act^{i-1})+l_Q, 
\qquad l_{i}(\act^{i-1},\rsign) =l_{i}(\act^{i-1})+l_q. 
\end{align*}
Note that $l_{i}(\act^{i-1},\Rsign)=2l_{i}(\act^{i-1}) -l_{i}(\act^{i-1},\Lsign)$ and $l_{i}(\act^{i-1},\rsign)=2l_{i}(\act^{i-1}) -l_{i}(\act^{i-1},\lsign)$. 

The expected utility of player $i$ with log likelihood ratio $l$ from action $\act_i=\ract$ if fraction $f$ of previous players chose $\ract$ is $\frac{\exp(l)}{1+\exp(l)}-f\cost$, but the expected utility from $\act_i=\lact$ is $\frac{1}{1+\exp(l)}-(1-f)\cost$. The payoff difference 
$
\Delta(l,f):= \frac{\exp(l)-1}{1+\exp(l)}+(1-2f)\cost
$ 
determines the best response: player $i$ chooses $\ract$ if $\Delta(l,f)>0$ and only if $\Delta(l,f)\geq 0$. 
Define the cutoff log likelihood ratio 
\begin{align}
\label{cutoffllr}
l_{\cost}(f):=\ln(1-\cost+2f\cost) -\ln(1+\cost-2f\cost)
\end{align} 
at which a player switches from action $\lact$ to $\ract$. Clearly $l_{\cost}(\frac{1}{2})=0$ and $l_{\cost}(1)=-l_{\cost}(0)$.\footnote{More generally, $l_{\cost}$ is antisymmetric around $\frac{1}{2}$, i.e.\ $l_{\cost}(f)=-l_{\cost}(\frac{1}{2}-f)$ for any $f\geq\frac{1}{2}$.
} 

The next section derives the optimal action choices of the players and provides sufficient conditions for herding to increase when players want to take a different action from their predecessors.

\section{Beliefs and best responses}
\label{sec:results}

Player $1$ chooses $\act_1=\lact$ after signals $\Lsign,\lsign$ and $\act_1=\ract$ after $\Rsign,\rsign$, due to the assumption $l_q>l_{0}$. There are no predecessors for player $1$, so the optimal action $\act_{1}^{*}$ does not depend on $\cost$. Similarly, if exactly half the predecessors of an odd-numbered player $2i-1$ choose action $\lact$, then $\cost$ does not affect $a_{2i-1}$. 

Given $l_q>l_{0}$, player $2$'s log likelihood ratios conditional on player $1$'s action $\act_{1}$ are $l_2(\lact) =l_0-l_{Qq}$ and $l_2(\ract) =l_0+l_{Qq}$ before observing $s_2$. The interpretation of $\act_{1}=\lact$ from player $2$'s perspective is as the `average' of the signals $\Lsign$ and $\lsign$, and similarly $\act_{1}=\ract$. 
If the congestion cost is not too large and the prior not too extreme, then the action of player $2$ responds to $s_2$. 
%
%
Lemma~\ref{lem:a1a2} characterises when the actions of the first two agents are informative. 
\begin{lemma}
\label{lem:a1a2}
Player $1$'s action is informative if $l_Q>l_{0}$ and only if $l_Q\geq l_{0}$. 
Player $2$'s action is informative after any $\act_1$ if $l_q>l_{0}$, $l_0-l_{Qq}-l_Q<l_{\cost}(0)$ and $l_0+l_{Qq}-l_Q<l_{\cost}(1)$. 
\end{lemma}
\begin{proof}
If $l_Q>l_{0}$, then $l_1(\Lsign)<0$, so $\act_1(\Lsign)=\lact$. Due to $l_0>0$, $\act_1(\rsign)=\act_1(\Rsign)=\ract$, thus $\act_1$ is informative. If $l_Q<l_{0}$, then $l_1(\Lsign)>0$, so $\act_1=\ract$ for any $s_1$. 

Clearly $\act_2(\lact,\Rsign)=\ract$ for any $\cost\geq0$. Before observing $s_2$, if $l_q>l_{0}$, then $l_2(\lact) =l_0-l_{Qq}$ and $l_2(\ract) =l_0+l_{Qq}$. Then $l_0-l_{Qq}-l_Q<l_{\cost}(0)$ implies $\Delta(l_2(\lact,\Lsign),f)<0$, so $\act_2(\lact,\Lsign)=\lact$, ensuring that $\act_2$ is informative after $\act_1=\lact$. 

The condition $l_0+l_{Qq}+l_Q>l_{\cost}(1)$ ensuring $\act_2(\ract,\Rsign)=\ract$ is implied by $l_0-l_{Qq}-l_Q<l_{\cost}(0)$ and $l_{\cost}(0)=-l_{\cost}(1)$. If $l_0+l_{Qq}-l_Q<l_{\cost}(1)$, then $\act_2(\ract,\Lsign)=\lact$, thus $\act_2$ is informative after $\act_1=\ract$. 
\end{proof}
The maintained assumption $l_q>l_{0}$ implies $l_Q> l_{0}$, which ensures $\act_1$ is informative by Lemma~\ref{lem:a1a2}. The conditions sufficient for $\act_2$ to be informative are not necessary. The interpretation of $l_0-l_{Qq}-l_Q+l_{\cost}(1)<0$ is that the congestion cost is small enough for player $2$ not to ignore own signal just to take a different action from player $1$. If $l_0+l_{Qq}-l_Q-l_{\cost}(1)<0$, then the prior probability of state $\rstate$ is low enough that a strong signal $s_2=\Lsign$ in favour of $\lstate$ together with the preference to differ outweighs the prior and player $1$'s action $\act_1=\ract$. 

Next, sufficient conditions are provided for herding to increase after the introduction of the desire to differ from previous agents. Increased herding means that actions become uninformative after some histories, but not the reverse. 
The set of histories after which herding occurs under $\cost>0$, but not under $\cost=0$ can have probability close to $1$, as the numerical example after Proposition~\ref{prop:herdincreases} demonstrates. 
Proposition~\ref{prop:herdincreases} proves increased herding for the first four players under $\cost>0$ compared to $\cost=0$. After that, Lemma~\ref{lem:a5} shows that player $5$ also herds more under $\cost>0$. 
\begin{prop}
\label{prop:herdincreases}
Assume $l_q>l_{0}$ and $l_0-l_{Qq}-l_Q+l_{\cost}(1)<0$. 
\\(a) If $\cost=0$, $l_0-l_{Qq}+l_q<0$ and $l_0+l_{Qq}+l_{\neg Q}-l_Q<0$, then $\act_3$ is informative after any $\act^{2}$. 
\\(b) If $\cost>0$, $l_0+l_{Qq}-l_q-l_{\cost}(1)<0$ and $l_0-2l_{Qq}+l_Q+l_{\cost}(1)<0$, then $\act_3$ is uninformative after $\act_1=\act_2$, the probability of which is $(Q+q)^2+(1-Q-q)^2> \frac{1}{2}$. If in addition $l_0+l_{Qq}-l_q-l_{\cost}\left(\frac{i+1}{2i+1}\right)<0$, then $\act_{2i+3}$ is uninformative after $\act_{2i+1}=\act_{2i+2}$. 
\\(c) If $l_0-l_{Qq}+l_{\neg Q}<0$ and $\act_4(\act^3,s_3)$ is informative under $\cost>0$, then also under $\cost=0$. 
\end{prop}
\begin{proof}
(a) The condition $l_0+l_{Qq}-l_Q-l_{\cost}(1)<0$ is implied by $l_0+l_{Qq}+l_{\neg Q}-l_Q<0$ and by $l_0+l_{Qq}-l_q-l_{\cost}(1)<0$. Recall $l_{\cost}(f)=-l_{\cost}(\frac{1}{2}-f)$. 

If $l_q>l_{0}$, then $\act_1(\Lsign)=\act_1(\lsign)=\lact$ and $\act_1(\Rsign)=\act_1(\rsign)=\ract$. In this case, $l_0-l_{Qq}-l_Q+l_{\cost}(1)<0$ ensures $\act_2(\lact,\Lsign)=\lact$. If $\cost=0$, $l_0-l_{Qq}+l_q<0$ and $l_0+l_{Qq}-l_Q-l_{\cost}(1)<0$, then player $2$'s actions are $\act_2(\ract,s_2\neq \Lsign) =\ract =\act_2(\lact,\Rsign)$ and $\act_2(\ract,\Lsign) =\lact =\act_2(\lact,s_2\neq \Rsign)$, so 
\begin{align*}
&l_3(\lact,\lact)=l_0-l_{Qq}-l_{\neg Q},\qquad l_3(\lact,\ract)=l_0-l_{Qq}+l_{Q}, 
\\&\notag l_3(\ract,\lact)=l_0+l_{Qq}-l_{Q},\qquad l_3(\ract,\ract)=l_0+l_{Qq}+l_{\neg Q}. 
\end{align*}

When $\cost=0$ and $l_0+l_{Qq}+l_{\neg Q}-l_Q<0$, player $3$'s action is informative after $\act_1=\act_2$: 
\begin{center}
\begin{tabular}{lrr}
\hline
private history $(\act^2,s_3)$ & $l_3(\act^2,s_3)$ & $\act_3(\act^2,s_3)$ \\
\hline
$\ract,\ract,\Lsign$ & $l_0+l_{Qq}+l_{\neg Q}-l_Q<0$ & $\lact$ \\
$\ract,\ract,s_3\neq \Lsign$ & $\geq l_0+l_{Qq}+l_{\neg Q}-l_{q}>0$ & $\ract$ \\
$\lact,\lact,\Rsign$ & $l_0-l_{Qq}-l_{\neg Q}+l_Q>0$ & $\ract$ \\
$\lact,\lact,s_3\in\set{\lsign,\Lsign}$ & $\leq l_0-l_{Qq}-l_{\neg Q}-l_{q}<0$ & $\lact$ \\
\hline
\end{tabular}
\end{center}

If $\cost=0$, then $\act_3$ is informative after $\act_1\neq \act_2=\ract$, because $l_3(\lact,\ract,\Lsign) =l_0-l_{Qq} <0$ due to $l_{Qq}>l_q>l_0$, and $l_3(\lact,\ract,\rsign) =l_0-l_{Qq}+l_Q+l_q >0$. Action $\act_3$ is informative after $\act_1\neq \act_2=\lact$, because $l_3(\ract,\lact,\Rsign) =l_0+l_{Qq} >0$ and $l_3(\ract,\lact,\Lsign) =l_0+l_{Qq}-2l_Q<0$ due to $l_Q>l_{Qq}>l_0$. 

(b) 
If $\cost>0$, $l_q>l_{0}$, $l_0+l_{Qq}-l_q-l_{\cost}(1)<0$ and $l_0-l_{Qq}+l_q+l_{\cost}(1)<0$, then $\act_2(\act_1,\Lsign) =\act_2(\act_1,\lsign) =\lact$ and $\act_2(\act_1,\Rsign) =\act_2(\act_1,\rsign) =\ract$ for any $\act_1$, so $l_3(\lact,\lact)=l_0-2l_{Qq}$, $l_3(\lact,\ract)=l_3(\ract,\lact)=l_0$ and $l_3(\ract,\ract)=l_0+2l_{Qq}$. 

If $\cost>0$ and $l_3(\lact,\lact,\Rsign) =l_0-2l_{Qq}+l_Q+l_{\cost}(1)<0$, then $l_3(\lact,\lact,s_3)<0$ for any $s_3$, so $\act_3(\lact,\lact,s_3) =\lact$ for any $s_3$. The condition $l_0-2l_{Qq}+l_Q+l_{\cost}(1)<0$ implies $l_0+2l_{Qq}-l_Q-l_{\cost}(1)>0$, so $l_3(\ract,\ract,s_3)>0$ and $\act_3(\ract,\ract,s_3) =\ract$ for any $s_3$. 
If player $3$ herds after $\act_2=\act_1$, then so do all subsequent players, because $f=1$ remains unchanged. More generally, if player $i$ herds after some history $\act^{i-1}$ and $|f-\frac{1}{2}|$ weakly decreases over time, then all subsequent players $j>i$ also herd after any continuation of $\act^{i-1}$. 

In the $\cost>0$ case, if $\act_1\neq \act_2$, $l_0+l_{Qq}-l_q-l_{\cost}(1)<0$ and $l_0-l_{Qq}+l_q+l_{\cost}(1)<0$, then $l_3(\act^{2})=l_0$ and $f=\frac{1}{2}$, so $\act_3$ is informative by $l_q>l_0$. For any period $j$, if $f=\frac{1}{2}$, then $j$ is odd, $l_{j}(\act^{j-1})=l_0$ and if $\act_{j+1}=\act_{j}$, then in period $j+2$, $f=\frac{(j-1)/2+2}{j+1}$. If $l_0-l_{Qq}+l_q+l_{\cost}(1)<0$, then $l_0-l_{Qq}+l_q+l_{\cost}(f)<0$ for any $f\leq 1$. 
Whenever $f=\frac{1}{2}$ and $l_{j}(\act^{j-1})=l_0$, the game essentially restarts, with player $j$ in the role of player $1$ and a reduced $l_{\cost}(f)$, because $f$ responds less to $\act_{j+1}=\act_{j}$. Therefore if a herd has not started after $\act^{2i}$ (which implies $\act_{2t}\neq \act_{2t-1}$ for all $t\leq i$), 
then it starts after $(\act^{2i},\lact,\lact)$, and if $l_0+l_{Qq}-l_q-l_{\cost}\left(\frac{i+1}{2i+1}\right)<0$, then also after $(\act^{2i},\ract,\ract)$. The conditional probability of a herd is $\Pr(\act_{2i+2}=\act_{2i+1}|\act_{2i+1}\neq \act_{2i}) =1-2(Q+q)(1-Q-q)=(Q+q)^2+(1-Q-q)^2$. 

(c) 
Table~\ref{tab:l4a4} displays $l_3(\act^3)$ in the cases $\cost=0$ and $\cost>0$, as well as the conditions for $\act_4(\act^3,s_3)$ to be informative. 
\begin{table}[tb!]
\begin{center}
\caption{Public log likelihood ratios of player $4$ before seeing $s_4$, and conditions under which $\act_4$ responds to $s_4$. 
Maintained assumptions: $l_0-l_{Qq}+l_q<0$, $l_0+l_{Qq}+l_{\neg Q}-l_Q<0$, $l_0+l_{Qq}-l_q-l_{\cost}(1)<0$ and $l_0-2l_{Qq}+l_Q+l_{\cost}(1)<0$.}
\label{tab:l4a4}
\begin{tabular}{ @{\extracolsep{6pt}} lllll @{} }
\hline 
 & \multicolumn{2}{c}{$l_4(\act^3)$} & \multicolumn{2}{c}{$\act_4(\act^3,s_4)$ responds to $s_4$ if} \\
 \cline{2-3} \cline{4-5}
history $\act^3$ & $\cost=0$ & $\cost>0$ & $\cost=0$ & $\cost>0$ \\
\hline
$\lact,\lact,\lact$ & $l_0-l_{Qq}-2l_{\neg Q}$ & $l_0-2l_{Qq}$ & $l_0-l_{Qq}-2l_{\neg Q}+l_Q>0$ & never  \\ 
$\ract,\ract,\ract$ & $l_0+l_{Qq}+2l_{\neg Q}$ & $l_0+2l_{Qq}$ & $l_0+l_{Qq}+2l_{\neg Q}-l_Q<0$ & never  \\ 
$\lact,\lact,\ract$ & $l_0-l_{Qq}-l_{\neg Q}+l_Q$ & off-path & always  & never  \\ 
$\ract,\ract,\lact$ & $l_0+l_{Qq}+l_{\neg Q}-l_Q$ & off-path  & always & never  \\ 
$\lact,\ract,\lact$ & $l_0-l_{Qq}+l_Q-l_Q$ &  $l_0-l_{Qq}$ & always &  always \\ 
$\ract,\lact,\ract$ & $l_0+l_{Qq}-l_Q+l_Q$ &  $l_0+l_{Qq}$  & always &  always \\ 
$\lact,\ract,\ract$  & $l_0-l_{Qq}+l_Q+l_{\neg Q}$ & $l_0+l_{Qq}$ & $l_0-l_{Qq}+l_{\neg Q}<0$ & always \\ 
$\ract,\lact,\lact$  & $l_0+l_{Qq}-l_Q-l_{\neg Q}$ & $l_0-l_{Qq}$ & always & always \\ 
\hline 
\end{tabular} 
\end{center}
\end{table}
Sufficient for $\act_4(\act^3,s_3)$ to be informative under $\cost=0$ is that $l_3(\act^3,\Lsign)<0<l_3(\act^3,\Rsign)$, which is how the fourth column of Table~\ref{tab:l4a4} is derived from the second. 
Under $\cost>0$, if $\act_1=\act_2$, then herding already started from $\act_3$, so $\act_4$ is uninformative. If $\act_1\neq\act_2$, then player $4$ faces the same decision problem as player $2$, so by Lemma~\ref{lem:a1a2}, $\act_4$ is informative for any $\act_3$. More generally, if $\act_{2t-1}\neq\act_{2t}$ for all $t<i$, then player $2i$ faces the same decision problem as player $2$, so by Lemma~\ref{lem:a1a2}, $\act_{2i}$ is informative for any $\act_{2i-1}$. Table~\ref{tab:l4a4} shows that if $\act_{4}$ is informative under $\cost>0$ and $l_0-l_{Qq}+l_{\neg Q}<0$, then $\act_{4}$ is informative under $\cost=0$. 
\end{proof}

Proposition~\ref{prop:herdincreases} is not vacuous---two numerical examples satisfying the assumptions are presented next. 
\begin{example}
\label{ex:1}
Take either $p_0=\frac{1}{2}$, $p_S=\frac{61}{64}$, $Q=\frac{15611}{16384}$, $q= \frac{9}{256}$ and any $\cost\in[0,\frac{1}{3}]$, or $p_0=\frac{5}{8}$, $p_S=\frac{61}{64}$, $Q=\frac{3903}{4096}$, $q=\frac{9}{256}$ and $\cost\approx 0.01$. 
In both cases, player $3$'s herding probability increases from $0$ to $(Q+q)^2+(1-Q-q)^2\approx 0.98$. Herding by player $4$ (and $5$ and $6$, as Lemmas~\ref{lem:a5},~\ref{lem:a6} below show) increases after every history. 
\end{example}

The intuition for the condition $l_q>l_0$ in Proposition~\ref{prop:herdincreases} is that player $1$'s weak signal outweighs the prior, so player $1$ always follows own signal. The assumption $l_0-l_{Qq}+l_q<0$ ensures that with $\cost=0$, the prior $p_{0}$ is close enough to $\frac{1}{2}$ for player $2$'s weak signal in favour of state $\rstate$ not to outweigh the ``average'' signal (which is player $1$'s action) favouring state $\lstate$. The best response of player $2$ is then to follow player $1$ except when $s_2$ is strong and disagrees with $\act_1$. 

From player $3$'s perspective, observing $\act_2\neq \act_1$ is equivalent to seeing a strong signal $s_2=\act_2$, but observing $\act_2=\act_1$ conflates the three other signals $\lsign,\rsign$ and $s_3\in\set{\Lsign,\Rsign}\setminus \set{\act_2}$, in which case $3$'s log likelihood ratio moves by only $l_{\neg Q}$. 
The intuition for $l_0+l_{Qq}+l_{\neg Q}-l_Q<0$ is that the effect $l_{Q}$ of a strong signal outweighs the combined prior $l_0$, average signal $l_{Qq}$ and the conflation $l_{\neg Q}$ of three signals when $\cost=0$. Thus player $3$ always follows a strong signal $s_3\in\set{\Lsign,\Rsign}$, regardless of whether $s_3\neq\act_2$, so $\act_3$ is informative. 

Under $\cost>0$, the assumption $l_0+l_{Qq}-l_q-l_{\cost}(1)<0$ ensures that player $2$ always follows own signal, because the desire to differ\footnote{
If $l_0-l_{Qq}-l_q< -l_{\cost}(1)$, which is implied by $l_0-2l_{Qq}+l_Q<-l_{\cost}(1)$ and $l_Q>l_{Qq}$, then player $2$ does not ignore $s_2$ just to ensure $\act_2\neq\act_1$, i.e.\ $\cost$ is small enough not to induce an anti-herding information cascade. 
} 
from player $1$ combines with the effect of a weak signal to outweigh the prior and the information derived from $\act_1$. This choice of player $2$ to follow $s_2$ increases the informativeness of $\act_2=\act_1$, but decreases that of $\act_2\neq\act_1$. The more informative event $\act_2=\act_1$ together with $l_0-2l_{Qq}+l_Q+l_{\cost}(1)<0$ induces player $3$ to herd, because even a strong signal plus the desire to differ $l_{\cost}(1)$ do not overcome the effect $2l_{Qq}$ of two ``average'' signals. 
If player $3$ herds after $\act_2=\act_1$, then so do all subsequent players, because they have the same signal strengths and desire to differ. 

The less informative $\act_2\neq\act_1$ under $\cost>0$ does not reduce player $3$'s herding, because even under $\cost=0$, player $3$ follows a strong signal after $\act_2\neq\act_1$. No additional assumptions are needed, because $\act_2\neq\act_1$ is either a strong signal (if $\cost=0$) or average (if $\cost>0$) favouring the opposite state to $\act_1$. The average signal from $\act_2\neq\act_1$ neutralises $\act_1$, so $l_q>l_0$ is sufficient for $\act_3$ to respond to even weak signals. The strong signal from $\act_2\neq\act_1$ under $\cost=0$ is neutralised by player $3$'s strong private signal $s_3\in\set{\Lsign,\Rsign}\setminus\set{\act_2}$, in which case $\act_3=\act_1$. On the other hand, if $s_3=\act_2$, then $\act_3=\act_2\neq\act_1$, so the action of player $3$ is informative in the $\cost=0$ case as well.

In Table~\ref{tab:l4a4}, $l_0+l_{Qq}+2l_{\neg Q}-l_Q<0$ on the line corresponding to history $\act_i=\ract$ is sufficient for $l_0-l_{Qq}-2l_{\neg Q}+l_Q>0$ on the $\act_i=\lact$ line. The intuition for these conditions is that a strong signal $s_5$ overwhelms the effect of an ``average'' signal from $\act_1$ plus two conflations ($\act_2$ and $\act_3$) of the three signals other than a strong one opposing $\act_1$. 
The condition $l_0+l_{Qq}-l_{\neg Q}>0$ for $\act_4(\ract,\lact,\lact,s_4)$ to respond to $s_4$ under $\cost=0$ always holds (so is omitted from the last line of Table~\ref{tab:l4a4}), because $l_0\geq0$ and $l_{Qq}>l_{\neg Q}$. 
The maintained assumption $l_0-l_{Qq}+l_q<0$ is logically independent of the condition $l_0-l_{Qq}+l_{\neg Q}<0$ ensuring an informative $\act_4(\lact,\ract,\ract,s_4)$ (penultimate line in Table~\ref{tab:l4a4}), because both $l_q>l_{\neg Q}$ and $l_q<l_{\neg Q}$ are possible. The reason why $l_0-l_{Qq}+l_{\neg Q}<0$ is sufficient for $\act_4(\lact,\ract,\ract,s_4)$ to respond to $s_4$ is that the strong signal from $\act_2\neq \act_1=\lact$ is cancelled by $s_4=\Lsign$, resulting in $l_4(\lact,\ract,\ract,\Lsign)<0$, but if $s_4\in\set{\rsign,\Rsign}$, then $l_4(\lact,\ract,\ract,s_4)>0$. 

The next lemma compares the informativeness of $\act_5$ under $\cost=0$ to the $\cost>0$ case. It complements Proposition~\ref{prop:herdincreases} by showing that in addition to the increased herding by the first four agents, player $5$ also responds less to signals. A similar result for player $6$ is subsequently derived in Lemma~\ref{lem:a6}. 
\begin{lemma}
\label{lem:a5}
If $l_0-l_{Qq}+l_{\neg Q}<0$ and $\act_5(\act^4,s_5)$ is informative under $\cost>0$, then $\act_5(\act^4,s_5)$ is also informative under $\cost=0$. 
\end{lemma}
\begin{proof}
%
History $\act^3=(\ract,\lact,\lact)$. 
Under $\cost=0$, if $l_4(\ract,\lact,\lact,\rsign)=l_0+l_{Qq}-l_Q-l_{\neg Q}+l_{q}<0$, then player $5$'s log likelihood ratios continuing from $\act^3=(\ract,\lact,\lact)$ are $l_5(\ract,\lact,\lact,\lact)=l_0+l_{Qq}-l_Q-l_{\neg Q}-l_{\neg Q}$ and $l_5(\ract,\lact,\lact,\ract)=l_0+l_{Qq}-l_Q-l_{\neg Q}+l_{Q}$, because player $4$ chooses $\lact$ after a weak signal $s_4=\rsign$. In this case, $\act_5$ responds to $s_5$ if $l_0+l_{Qq}-2l_{\neg Q}>0$, because $l_5(\ract,\lact,\lact,\lact,\Rsign)=l_0+l_{Qq}-2l_{\neg Q}$ and $l_5(\ract,\lact,\lact,\lact,\Lsign)=l_0+l_{Qq}-2l_Q-2l_{\neg Q}<0$. After $(\ract,\lact,\lact,\ract)$, player $5$'s action always responds to the private signal. By comparison, recall that when $\cost>0$, player $5$ (and any odd player) herds if the preceding two players took the same action.

On the other hand, if $l_4(\ract,\lact,\lact,\rsign)>0$, then $l_5(\ract,\lact,\lact,\lact)=l_0+l_{Qq}-l_Q-l_{\neg Q}-l_{Qq}$ and $l_5(\ract,\lact,\lact,\ract)=l_0+l_{Qq}-l_Q-l_{\neg Q}+l_{Qq}$, because player $4$ chooses $\lact$ after a weak signal $s_4=\rsign$. In this case, $\act_5$ responds to $s_5$ if $l_5(\ract,\lact,\lact,\lact,\Rsign)=l_0-l_{\neg Q}>0$ (again, $\act_5$ always responds if the $\act^4$ contains an equal number of $\lact,\ract$). The condition $l_0-l_{\neg Q}>0$ fails in Example~\ref{ex:1} above, so player $5$ herds. When $\cost>0$, player $5$ always herds after history $(\ract,\lact,\lact,\lact)$. 

History $\act^3=(\lact,\ract,\ract)$. 
If $l_0-l_{Qq}+l_Q+l_{\neg Q}-l_q>0$, then $l_5(\lact,\ract,\ract,\lact)=l_0-l_{Qq}+l_Q+l_{\neg Q}-l_{Q}$ and $l_5(\lact,\ract,\ract,\ract)=l_0-l_{Qq}+l_Q+l_{\neg Q}+l_{\neg Q}$, because $\act_4(\lact,\ract,\ract,\lsign)=\ract$. The condition for $\act_5$ to respond to $s_5$ is $l_5(\lact,\ract,\ract,\ract,\Lsign)=l_0-l_{Qq}+l_Q+l_{\neg Q}+l_{\neg Q}-l_Q<0$, the same as for $\act_4$ to be informative after $\act^3=(\lact,\ract,\ract)$. 

In contrast, if $l_0-l_{Qq}+l_Q+l_{\neg Q}-l_q<0$, then $l_5(\lact,\ract,\ract,\lact)=l_0-l_{Qq}+l_Q+l_{\neg Q}-l_{Qq}$ and $l_5(\lact,\ract,\ract,\ract)=l_0-l_{Qq}+l_Q+l_{\neg Q}+l_{Qq}$. Action $\act_5$ is always informative after $\lact,\ract,\ract,\lact$, but never after $\lact,\ract,\ract,\ract$ (just like with $\cost>0$), because $l_5(\lact,\ract,\ract,\ract,\Lsign) =l_0+l_{\neg Q}>0$. 

Histories $\act^3=(\ract,\lact,\ract)$ and $(\lact,\ract,\lact)$ lead to $l_4(\act^3)=l_0\pm l_{Qq}$, so player $4$ faces the same decision problem as player $2$. Thus continuing from these histories, any player herds more under $\cost>0$ than under $\cost=0$. 
%
For histories in the top half of Table~\ref{tab:l4a4}, herding has already started with player $3$, so all subsequent players unambiguously herd more under $\cost>0$. 
\end{proof}

Lemma~\ref{lem:a6} compares the informativeness of the action $\act_6$ of player $6$ under $\cost=0$ and $\cost>0$, analogously to Lemma~\ref{lem:a5} for $\act_5$. 
\begin{lemma}
\label{lem:a6}
If $l_0+2l_{Qq}-l_{\neg Q}-l_Q<0$ and $\act_6(\act^5,s_6)$ is informative under $\cost>0$, then $\act_6(\act^5,s_6)$ is also informative under $\cost=0$. 
\end{lemma}
\begin{proof}
Based on Proposition~\ref{prop:herdincreases} and Lemma~\ref{lem:a5}, the only histories continuing from which player $6$ could conceivably herd more under $\cost=0$ are $\act^4=(\lact,\ract,\ract,\lact)$ and $(\ract,\lact,\lact,\ract)$. 
In these continuations, under $\cost>0$, player $6$ faces the same decision as player $2$, but this need not be the case under $\cost=0$. 
Consider first $\act^4=(\ract,\lact,\lact,\ract)$. Separate two cases based on the sign of $l_4(\ract,\lact,\lact,\rsign)=l_0+l_{Qq}-l_Q-l_{\neg Q}+l_{q}$. 

If $l_0+l_{Qq}-l_Q-l_{\neg Q}+l_{q}<0$, then $l_5(\ract,\lact,\lact,\ract)=l_0+l_{Qq}-l_{\neg Q}$. In this case, if $l_0+l_{Qq}-l_{\neg Q}-l_q<0$, then $l_6(\ract,\lact,\lact,\ract,\lact)=l_0-l_{\neg Q}$ (so $\act_6$ is informative) and $l_6(\ract,\lact,\lact,\ract,\ract)=l_0+2l_{Qq}-l_{\neg Q}$. Then $\act_6$ is informative if $l_0+2l_{Qq}-l_{\neg Q}-l_Q<0$. 

The other case given $l_0+l_{Qq}-l_Q-l_{\neg Q}+l_{q}<0$ is $l_0+l_{Qq}-l_{\neg Q}-l_q>0$, which implies $l_6(\ract,\lact,\lact,\ract,\lact)=l_0+l_{Qq}-l_{\neg Q}-l_Q$ and $l_6(\ract,\lact,\lact,\ract,\ract)=l_0+l_{Qq}-l_{\neg Q}+l_{\neg Q}$, for both of which, $\act_6$ is informative. 

If $l_0+l_{Qq}-l_Q-l_{\neg Q}+l_{q}>0$, then $l_5(\ract,\lact,\lact,\ract)=l_0+2l_{Qq}-l_Q-l_{\neg Q}$. In this case, if $l_0+2l_{Qq}-l_Q-l_{\neg Q}-l_q<0$ (implied by $l_0+2l_{Qq}-l_Q-l_{\neg Q}<0$) and $l_0+2l_{Qq}-l_Q-l_{\neg Q}+l_q>0$ (implied by $l_0+l_{Qq}-l_Q-l_{\neg Q}+l_{q}>0$), then $l_6(\ract,\lact,\lact,\ract,\lact)=l_0+l_{Qq}-l_Q-l_{\neg Q}$ (so $\act_6$ is informative) and $l_6(\ract,\lact,\lact,\ract,\ract)=l_0+3l_{Qq}-l_Q-l_{\neg Q}$. Therefore if $l_0+3l_{Qq}-2l_Q-l_{\neg Q}<0$ (which is implied by $l_0+2l_{Qq}-l_Q-l_{\neg Q}<0$), then $\act_6$ is informative. 

Consider next $\act^4=(\lact,\ract,\ract,\lact)$, so $l_4(\lact,\ract,\ract,\lsign)=l_0-l_{Qq}+l_Q+l_{\neg Q}-l_{q}$. 
If $l_0-l_{Qq}+l_Q+l_{\neg Q}-l_{q}>0$, then $l_5(\lact,\ract,\ract,\lact)=l_0-l_{Qq}+l_{\neg Q}$. In this case, if $l_0-l_{Qq}+l_{\neg Q}+l_q>0$, then $l_6(\lact,\ract,\ract,\lact,\ract)=l_0+l_{\neg Q}$ (so $\act_6$ is informative) and $l_6(\lact,\ract,\ract,\lact,\lact)=l_0-2l_{Qq}+l_{\neg Q}$. Then $\act_6$ is informative if $l_0-2l_{Qq}+l_{\neg Q}+l_Q>0$, sufficient for which is $l_0+2l_{Qq}-l_{\neg Q}-l_Q<0$. 

The other case given $l_0-l_{Qq}+l_Q+l_{\neg Q}-l_{q}>0$ is $l_0-l_{Qq}+l_{\neg Q}+l_q<0$, which implies $l_6(\lact,\ract,\ract,\lact,\ract)=l_0-l_{Qq}+l_{\neg Q}+l_Q$ and $l_6(\lact,\ract,\ract,\lact,\lact)=l_0-l_{Qq}+l_{\neg Q}-l_{\neg Q}$, so $\act_6(\lact,\ract,\ract,\lact,\lact,s_6)$ is informative. 
Action $\act_6(\lact,\ract,\ract,\lact,\ract,s_6)$ is informative if $l_0-l_{Qq}+l_{\neg Q}<0$, which is implied by $l_0-l_{Qq}+l_{\neg Q}+l_q<0$. 

If $l_0-l_{Qq}+l_Q+l_{\neg Q}-l_{q}<0$, then $l_5(\lact,\ract,\ract,\lact)=l_0-2l_{Qq}+l_Q+l_{\neg Q}$. In this case, if $l_0-2l_{Qq}+l_Q+l_{\neg Q}+l_q>0$ (which is implied by $l_0+2l_{Qq}-l_{\neg Q}-l_Q<0$) and $l_0-2l_{Qq}+l_Q+l_{\neg Q}-l_q<0$ (implied by $l_0-l_{Qq}+l_Q+l_{\neg Q}-l_{q}<0$), then $l_6(\lact,\ract,\ract,\lact,\ract)=l_0-l_{Qq}+l_Q+l_{\neg Q}$, so $\act_6(\lact,\ract,\ract,\lact,\ract,s_6)$ is informative, because $l_0-l_{Qq}+l_Q+l_{\neg Q}-l_{q}<0$ implies $l_0-l_{Qq}+l_{\neg Q}<0$. 
Also, $l_6(\lact,\ract,\ract,\lact,\lact)=l_0-3l_{Qq}+l_Q+l_{\neg Q}$, thus if $l_0-3l_{Qq}+2l_Q+l_{\neg Q}>0$ (sufficient for which is $l_0+2l_{Qq}-l_{\neg Q}-l_Q<0$), then $\act_6$ is informative. 
\end{proof}
The assumption $l_0+2l_{Qq}-l_{\neg Q}-l_Q<0$ in Lemma~\ref{lem:a6} that suffices for player $6$ to herd more under $\cost>0$ is satisfied in the Example~\ref{ex:1} above. Therefore the set of histories in which the first six players herd under $\cost>0$ is a proper superset of the histories in which they herd under $\cost=0$. 

In some long enough histories, the probability of herding under $\cost=0$ may overtake that under $\cost>0$. This is because the effective congestion cost decreases when $f$ approaches $\frac{1}{2}$, which occurs each time the history lengthens by two actions without a herd having started. \cite{eyster+2014} show that as the congestion cost approaches zero, learning increases in the limit as time goes to infinity. The comparison of the limiting probabilities of learning under $\cost=0$ and $\cost>0$ is complicated, because it depends on the speed of convergence of $|f-\frac{1}{2}|$ in the present paper relative to that of the congestion cost in \cite{eyster+2014}. What is clear is that discounting the benefit of agents in the far future learning makes the welfare impact of the large initial increase in herding under $\cost>0$ overwhelm any eventual overtaking of the learning probability under $\cost=0$. In other words, the discounted probability of correct decisions is significantly smaller when there is a desire to differ from previous movers.

\section{Discussion}
\label{sec:discussion}

The result that herding may increase with the desire to differ from previous movers is robust to varying the informativeness of signals or the congestion cost within some bounds. The informativeness and cost may also differ to some extent across players. Unboundedly informative signals or a strong enough preference for non-conformity break herding, as established in the previous literature. If the congestion cost is small enough, then it does not affect players' actions, because it does not outweigh the weakest of the finitely many signals. 

In some applications, the congestion cost depends only on the actions of some preceding agents, not all. For example, if a service provider is capacity constrained and can serve only $m$ agents at a time or finishes the service in at most $m$ periods, then an agent's payoff only depends on the choices of the $m$ immediate predecessors. The desire to differ may increase conformity also in this case, as is clear from redefining $f$ in Section~\ref{sec:results} to be the fraction of agents among the preceding $m$ who choose $\ract$. 

Even if congestion depends only on the immediately preceding agent, a more informative $\act_2=\act_1$ can motivate player $3$ to herd. The less informative $\act_2\neq\act_1$ cannot reduce player $3$'s herding compared to the $\cost=0$ case if $3$ does not herd after $\act_2\neq\act_1$ under $\cost=0$. Thus the overall probability of herding may increase, as in the baseline model. The proofs simplify, because each time the belief returns to the prior, the subgame is identical to the whole game. In particular, the condition $l_0+l_{Qq}-l_q-l_{\cost}\left(\frac{i+1}{2i+1}\right)<0$ in Proposition~\ref{prop:herdincreases} for herding to start in period $2i+3$ conditional on not having started earlier may be omitted w.l.o.g., because it reduces to $l_0+l_{Qq}-l_q-l_{\cost}\left(1\right)<0$. 

Qualitatively similar results also obtain if congestion depends on a discounted (or otherwise weighted) average of the actions of previous movers. Again, player $3$ faces the same problem as in the baseline model, and the problems of subsequent odd players only differ in the effective congestion cost. 


\appendix
\section{Herding reduced by the desire to conform}
\label{sec:conform}

This section shows that a preference to match the actions of preceding agents may in fact reduce herding. The idea is similar to why the desire to differ may increase herding---the actions of previous players become more informative after some histories, less after others. In the current section, it is the less informative actions that matter. A strong signal overwhelms the effect of two previous less informative actions plus the desire to conform, but does not outweigh the more informative actions in the absence of a preference to follow previous movers. 

Only the differences from the setup in Section~\ref{sec:model} are mentioned. 
Payoffs are 
\begin{align*} 
u_{i}(\act^{i},\theta) =\text{\textbf{1}}\set{\act_{i} =\theta} {\color{red}+}\frac{\cost}{i-1}\sum_{j=1}^{i-1}\text{\textbf{1}}\set{\act_{j} =\act_{i}},
\end{align*} 
where $\cost\geq0$ as before, but here the payoff from an action 
increases in the fraction $f$ of previous agents taking that action. 

There are six possible signal realisations $s_i\in\set{\Lsign,\lsign,\losign,\rosign,\rsign,\Rsign}$, with $\lsign,\rsign$ interpreted as medium strength and $\losign,\rosign$ as weak. Signals $\Lsign,\lsign,\losign$ favour state $\lstate$, the others $\rstate$. 
The respective unconditional probabilities of a strong, medium and weak signal are $p_S:=\Pr(\Lsign)+\Pr(\Rsign)$, $p_s:=\Pr(\lsign)+\Pr(\lsign)$ and $p_{\sigma}:=\Pr(\losign)+\Pr(\rosign)$. The conditional probabilities are $\Pr(\Lsign|\lstate)=\Pr(\Rsign|\rstate)=:Q$, $\Pr(\lsign|\lstate)=\Pr(\rsign|\rstate)=:q$ and $\Pr(\losign|\lstate)=\Pr(\rosign|\rstate)=:\eta$. Assume $\frac{1}{2}<\frac{\eta}{p_{\sigma}}<\frac{q}{p_s}<\frac{Q}{p_S}<1$, which justifies the interpretations of the signals. 
Define 
\begin{description}
\item $l_q:=\ln q-\ln (p_s-Q) \in(0,l_Q)$, 
\item $l_{\eta}:=\ln \eta-\ln (p_{\sigma}-\eta) \in(0,l_q)$, 
\item $l_{Qq}:= \ln(Q+q)-\ln(p_S+p_s-Q-q)\in(l_q,l_Q)$, 
\item $l_{Qq\eta}:= \ln(Q+q+\eta)-\ln(1-Q-q-\eta)\in(l_{\eta},l_Q)$,
\item $l_{\neg qQ}:= \ln(p_{\sigma}+q+Q)-\ln(1-q-Q) \in(0,l_{Qq\eta})$,
\item $l_{\neg Q}:= \ln(p_s+p_{\sigma}+Q)-\ln(1-Q)\in(0,l_{\neg qQ})$. 
\end{description}

The next result is analogous to Proposition~\ref{prop:herdincreases} and provides sufficient conditions for herding to decrease when conformism is introduced. 
\begin{prop}
\label{prop:conform}
Assume $l_{\eta}>l_{0}$. 
If $\cost=0$, $l_0+l_{Qq\eta}-l_q<0$, $l_0-l_{Qq\eta}+l_{\eta}<0$ and $l_0-l_{Qq\eta}-l_{\neg qQ}+l_{Q}<0$, then $\act_3$ is uninformative after $\act_{2}=\act_1$, the probability of which is $(Q+q+\eta)(p_{\sigma}+q+Q)+(1-Q-q-\eta)(1-q-Q)$. 
\\ If $\cost>0$, $l_0-l_{Qq\eta}+l_q-l_{\cost}(1)<0$ and $l_0+l_{Qq\eta}+l_{\neg Q}-l_{Q}+l_{\cost}(1)<0$, then $\act_3$ is informative after any history. 
\end{prop}
\begin{proof}
The assumption $l_0<l_{\eta}$ ensures that player $1$ follows own signal. Then player $2$'s public log likelihood ratios are $l_2(\lact)=l_0-l_{Qq\eta}$ and $l_2(\ract)=l_0+l_{Qq\eta}$. 

$\cost=0$. 
Assume $l_0+l_{Qq\eta}-l_q<0$ and $l_0-l_{Qq\eta}+l_{\eta}<0$, so $\act_2(\ract,\lsign)=\act_2(\lact,\rosign)=\lact$ and by implication, $\act_2(\lact,\rsign)=\act_2(\ract,\losign)=\ract$. 
Player $3$'s log likelihood ratios before seeing $s_3$ are 
\begin{align*}
&l_3(\lact,\lact) =l_0-l_{Qq\eta}-l_{\neg qQ}, \qquad l_3(\lact,\ract) =l_0-l_{Qq\eta}+l_{Qq}, 
\\&\notag l_3(\ract,\lact) =l_0+l_{Qq\eta}-l_{Qq}, \qquad l_3(\ract,\ract) =l_0+l_{Qq\eta}+l_{\neg qQ}. 
\end{align*}
If $l_0-l_{Qq\eta}-l_{\neg qQ}+l_{Q}<0$, then $\act_3$ is uninformative after $\act_2=\act_1$,  i.e.\ a herd starts. After $\act_2\neq \act_1$, player $3$'s action always responds to signals. 

$\cost>0$. 
If $l_0-l_{Qq\eta}+l_q-l_{\cost}(1)<0$ and $l_0+l_{Qq\eta}-l_{Q}+l_{\cost}(1)<0$ (which is implied by $l_0+l_{Qq\eta}+l_{\neg Q}-l_{Q}+l_{\cost}(1)<0$), then $\act_2(\lact,\rsign)=\act_2(\ract,\Lsign)=\lact$ and by implication, $\act_2(\ract,\lsign)=\act_2(\lact,\Rsign)=\ract$. 
Player $3$'s log likelihood ratios before observing $s_3$ are then 
\begin{align*}
&l_3(\lact,\lact) =l_0-l_{Qq\eta}-l_{\neg Q}, \qquad l_3(\lact,\ract) =l_0-l_{Qq\eta}+l_{Q}, 
\\&\notag l_3(\ract,\lact) =l_0+l_{Qq\eta}-l_{Q}, \qquad l_3(\ract,\ract) =l_0+l_{Qq\eta}+l_{\neg Q}. 
\end{align*}
If $l_0+l_{Qq\eta}+l_{\neg Q}-l_{Q}+l_{\cost}(1)<0$, then a strong signal switches the sign of $l_3(\ract,\ract)$, so $\act_3$ is informative after $\act_2=\act_1=\ract$ and by implication after any history. 
\end{proof}
The next example exhibits parameter values satisfying the assumptions of Proposition~\ref{prop:conform}.
\begin{example} 
Let $l_0=0$, $p_S=\frac{4}{5}$, $p_{s}=p_{\sigma}=\frac{1}{10}$, $\frac{Q}{p_S}\approx 0.984$, $\frac{q}{p_s}\approx 0.93$, $\frac{\eta}{p_{\sigma}}\approx 0.5$ and $\cost\approx 1.9\cdot 10^{-6}$, or alternatively $p_0=0.51$, $p_S=\frac{4}{5}$, $p_{s}=p_{\sigma}=\frac{1}{10}$, $\frac{Q}{p_S}\approx 0.987$, $\frac{q}{p_s}\approx 0.935$, $\frac{\eta}{p_{\sigma}}\approx 0.5002$ and $\cost\approx 0.04$. 

The probability of $\act_2=\act_1$ is $(Q+q+\eta)(p_{\sigma}+q+Q)+(1-Q-q-\eta)(1-q-Q)\approx 0.92$ under $\cost=0$, but 
$(Q+q+\eta)(p_{\sigma}+p_s+Q)+(1-Q-q-\eta)(1-Q)\approx 0.94$ under $\cost>0$. 
Thus in both examples, the probability that the action of player $3$ is informative rises from about $0.08$ to $1$ when the desire to conform is introduced. 
\end{example}


\bibliographystyle{ecta}
\bibliography{teooriaPaberid} 
\end{document}